\documentclass[letterpaper]{article} 
\usepackage{aaai25}  
\usepackage{times}  
\usepackage{helvet}  
\usepackage{courier}  
\usepackage[hyphens]{url}  
\usepackage{graphicx} 
\usepackage{natbib}  
\usepackage{caption} 
\usepackage{tabularx}
\usepackage{xcolor}
\usepackage[T1]{fontenc}
\usepackage{currvita}

\usepackage{algorithm}
\usepackage{algorithmic}
\usepackage{amsmath}

\usepackage{newfloat}
\usepackage{listings}
\DeclareCaptionStyle{ruled}{labelfont=normalfont,labelsep=colon,strut=off} 
\floatstyle{ruled}
\newfloat{listing}{tb}{lst}{}
\floatname{listing}{Listing}
\title{Quantified Linear and Polynomial Arithmetic Satisfiability \\ via Template-based Skolemization \thanks{This work was partially funded by ERC CoG 863818 (ForM-SMArt) and Austrian Science Fund (FWF) 10.55776/COE12}}
\author{
Krishnendu Chatterjee\textsuperscript{\rm 1},
Ehsan Kafshdar Goharshady\textsuperscript{\rm 1},
Mehrdad Karrabi\textsuperscript{\rm 1},
Harshit J Motwani\textsuperscript{\rm 2},
Maximilian Seeliger\textsuperscript{\rm 3},
\DJ or\dj e $\mathbf{\check{\text{Z}}}$ikeli\'{c}\textsuperscript{\rm 4}
}
\affiliations{
\textsuperscript{\rm 1}Institute of Science and Technology Austria, Klosterneuburg, Austria\\
\textsuperscript{\rm 2}Hong Kong University of Science and Technology, Clear Water Bay, Hong Kong\\
\textsuperscript{\rm 3}Vienna University of Technology, Vienna, Austria\\
\textsuperscript{\rm 4}Singapore Management University, Singapore, Singapore\\

ehsan.goharshady@ist.ac.at
}

\usepackage{bibentry}

\usepackage{amsmath}
\usepackage{amsthm}
\usepackage{booktabs}
\usepackage{mathtools, amssymb}
\usepackage{paralist}
\usepackage{comment}
\usepackage{subfigure}

\newtheorem*{example*}{Example}
\newtheorem*{remark*}{Remark}
\newtheorem*{theorem*}{Theorem}

\newtheorem{definition}{Definition}

\newtheorem{theorem}{Theorem}

\newcommand{\sat}{\textsl{SAT}}
\newcommand{\monoid}{\textsl{SG}}

\newcounter{todocounter}

\newcommand{\eat}[1]{}

\renewcommand{\paragraph}[1]{\smallskip\noindent\emph{\textbf{#1.}}}

\newcommand{\quantifier}{\mathcal{Q}}

\begin{document}

\maketitle

\begin{abstract}
	The problem of checking satisfiability of linear real arithmetic (LRA) and non-linear real arithmetic (NRA) formulas has broad applications, in particular, they are at the heart of logic-related applications such as logic for artificial intelligence, program analysis, etc. While there has been much work on checking satisfiability of unquantified LRA and NRA formulas, the problem of checking satisfiability of quantified LRA and NRA formulas remains a significant challenge. The main bottleneck in the existing methods is a computationally expensive quantifier elimination step. In this work, we propose a novel method for efficient quantifier elimination in quantified LRA and NRA formulas. We propose a template-based Skolemization approach, where we automatically synthesize linear/polynomial Skolem functions in order to eliminate quantifiers in the formula. The key technical ingredients in our approach are Positivstellens\"atze theorems from algebraic geometry, which allow for an efficient manipulation of polynomial inequalities. Our method offers a range of appealing theoretical properties combined with a strong practical performance. On the theory side, our method is sound, semi-complete, and runs in subexponential time and polynomial space, as opposed to existing sound and complete quantifier elimination methods that run in doubly-exponential time and at least exponential space. On the practical side, our experiments show superior performance compared to state-of-the-art SMT solvers in terms of the number of solved instances and runtime, both on LRA and on NRA benchmarks.
	
\end{abstract}

\section{Introduction}\label{sec:intro}

Satisfiability checking for logical formulas is one of the fundamental problems in automated reasoning that arises in many applications, including logic for artificial intelligence~\cite{AlurHK02}, planning and scheduling~\cite{KautzS92} or deductive verification~\cite{KroeningS16}. Satisfiability modulo theories (SMT) solvers have achieved impressive results at tackling this problem. However, while SMT solvers excel at checking satisfiability of quantifier-free formulas, many applications require reasoning about {\em quantified formulas}. For instance, automated reasoning problems in decision making for multi-agent systems typically involve quantified formulas, with quantifier alternation corresponding to the choices made by each agent~\cite{AlurHK02,ChatterjeeHP10}. Quantified formulas also commonly arise in program verification~\cite{GulwaniSV08} and synthesis~\cite{SolarLezamaTBSS06}.

The key challenge that arises in checking satisfiability of quantified formulas is the highly expensive {\em quantifier elimination} step, that SMT solvers need to perform in order to reduce the problem of checking satisfiability of a quantified formula to that of a quantifier-free formula. It is a classical result that, both in linear real arithmetic and in non-linear real arithmetic theories, the runtime complexity of quantifier elimination is doubly-exponential~\cite{weispfenning1988complexity}, making it practically infeasible and often a daunting task for SMT solvers. To overcome this challenge, some SMT solvers avoid the expensive quantifier elimination step by implementing methods such as quantifier instantiation~\cite{ReynoldsTGKDB13,MouraB07,GeBT09}. However, while being sound, quantifier instantiation is incomplete and may lead to ''Unknown'' outputs by the SMT solver. Another approach to satisfiability checking that avoids the quantifier elimination step is to consider the game semantics of quantified first-order formulas and to treat them as two-player games~\cite{BjornerJ15,farzan2016linear,MurphyK24}. However, these methods are only applicable to formulas in linear real arithmetic. 

Besides checking satisfiability of quantified formulas, another fundamental question is to obtain {\em witnesses of satisfiability for existentially quantified variables}. This feature is important in many applications. In decision making for multi-agent systems, the witnesses for existentially quantified variables gives rise to a strategy of the existentially quantified agent~\cite{AlurHK02,ChatterjeeHP10}. In planning and scheduling, the witnesses gives rise to a control policy~\cite{KautzS92}. Finally, in program synthesis applications, the witnesses gives rise to a program that satisfies the desired specification~\cite{SolarLezamaTBSS06}.


\paragraph{Our Approach} In this work, we propose a novel method for checking satisfiability of {\em quantified formulas} in {\em linear real arithmetic (LRA)} and {\em non-linear real arithmetic (NRA)}. Rather than following the approaches discussed above that try to sidestep quantifier elimination, at the core of our approach lies a {\em novel method for efficient quantifier elimination} in LRA and NRA. As mentioned above, sound and complete procedures for quantifier elimination are computationally expensive and inherently lead to doubly-exponential time and exponential space complexity~\cite{weispfenning1988complexity,brown2007complexity}. Hence, to overcome this complexity barrier, we focus on methods that are sound and {\em semi-complete}. While we defer the formal definition of semi-completeness to later parts of the paper, this relaxed notion of completeness intuitively means that the method is guaranteed to prove or disprove satisfiability whenever it can be witnessed by a certificate of a certain parametrized form. Our novel quantifier elimination procedure gives rise to a method for checking satisfiability of quantified formulas in LRA and NRA that is {\em sound}, {\em semi-complete}, and runs in {\em subexponential time} and {\em polynomial space}. To the best of our knowledge, this is the first method that provides all three of these desirable features. Our method also computes {\em witnesses of satisfiability for existentially quantified variables}.


\paragraph{Method Outline} Our method assumes that a quantified formula is provided in the prenex normal form
\begin{equation*}
	\phi \,\, ::= \,\, \quantifier_1 x_1.\quantifier_2x_2.\dots\quantifier_n x_n. p,
\end{equation*}
where $\quantifier_1, \quantifier_2,\dots,\quantifier_n \in \{\exists,\forall\}$ are quantifiers and $p$ is a quantifier-free formula. 
The method proceeds in three steps to check if $\phi$ is satisfiable (alternatively, to check if $\phi$ is not satisfiable, we may equivalently check if $\neg\phi$ is satisfiable):
\begin{compactenum}
	\item {\em Existential Quantifier Elimination via Skolemization.} Our method first removes all existential quantifiers from the formula $\phi$, by replacing each existentially quantified variable $x_i$ with a function over those variables in $x_1,\dots,x_{i-1}$ that are universally quantified. This process is called {\em Skolemization}, and functions used to express the existentially quantified variables via the universally quantified ones are called {\em Skolem functions}~\cite{scowcroft1988note}. To search for Skolem functions, our method follows what we call a {\em template-based Skolemization approach}, where it fixes a template for Skolem function of each existentially quantified $x_i$ in the form of a symbolic polynomial expression over universally quantified variables in $x_1,\dots,x_{i-1}$. At this stage, the polynomial coefficients are symbolic, and the concrete values of coefficients will be computed in later steps. Computing Skolem functions corresponds to computing witnesses of satisfiability for existentially quantified variables.
	\item {\em Universal Quantifier Elimination via Positivstellens\"atze.} Next, our method removes all universal quantifiers from the formula $\phi$. This is achieved by using Farkas' lemma~\cite{farkas1902theorie} and Positivstellens\"atze theorems~\cite{handelman1988representing,putinar1993positive,krivine1964anneaux,stengle1974nullstellensatz} from algebraic geometry. The procedure results in a quantifier-free formula $\phi^{\textsc{free}}$ whose satisfiability also implies the satisfiability of the formula~$\phi$.
	\item {\em Quantifier-free Formula Satisfiability Checking.} Finally, our method tries to prove that $\phi$ is satisfiable by proving that $\phi^{\textsc{free}}$ is satisfiable. This can be realized by using an off-the-shelf SMT solver, since SMT solvers already excel at satisfiability checking for quantifier-free formulas. If $\phi^{\textsc{free}}$  is proved to be satisfiable, then we conclude that $\phi$ is satisfiable. Otherwise, our method repeats Steps~1-3 for $\neg\phi$ to try to prove that $\phi$ is not satisfiable.
\end{compactenum}
We implement our method for satisfiability checking and experimentally compare it against Z3~\cite{MouraB08} and CVC5~\cite{cvc5}, which are both state-of-the-art SMT solvers. Our experiments show that, when required to find witnesses for existentially quantified variables, our method is able to solve a considerably larger number of quantified formula instances at lower average runtimes, both in LRA and in NRA. Moreover, we observe a large number of {\em unique} proofs for examples that could not be handled by neither Z3 nor CVC5. Thus, our method provides a significant step forward in tackling satisfiability checking and witness construction for quantified formulas, in LRA and NRA.

\paragraph{Contributions} Our contributions are as follows:
\begin{compactenum}
	\item {\em New Method for Quantifier Elimination.} We present a new method for efficient quantifier elimination in quantified formulas in LRA and NRA. Our method is based on a novel template-based Skolemization approach.
	\item {\em Efficient Satisfiability Checking.} Based on the above, we design a new method for efficient satisfiability checking for quantified formulas in LRA and NRA. Our method is sound, semi-complete and runs in subexponential time and polynomial space, parametrized by the size of Skolem function templates. To the best of our knowledge, this is the first method that provides all three of these desirable features for LRA and NRA quantified formulas. In contrast, previous sound and complete procedures have doubly-exponential time complexity and at least exponential space complexity. Moreover, our method also produces witnesses of satisfiability for the existentially quantified variables in the quantified formula.
	\item {\em Experimental Evaluation.} Our experiments showcase a strong practical performance of our method and a considerable improvement in the number of successful satisfiability checks, runtime, as well as unique satisfiability checks over two state-of-the-art SMT solvers.
\end{compactenum}


\section{Related Work}\label{sec:relatedwork}

The key step in existential quantifier elimination within most SMT solvers and computer-algebra systems is the so-called projection process. For LRA, this projection typically relies on Fourier-Motzkin elimination \cite{dantzig1972fourier}. For NRA, cylindrical algebraic decomposition (CAD) \cite{collins-cad75} is commonly employed. Although these methods are sound and complete, they suffer from a doubly exponential runtime complexity in the number of formula variables. To mitigate this issue, modern SMT solvers incorporate various heuristics and algorithms. Notably, a DPLL-style approach for quantified formulas in both LRA and NRA has been proposed \cite{jovanovic2013solving,de2013model}. This approach generalizes the CDCL method used in SAT solvers to handle first-order logic formulas. However, these methods ultimately depend on Fourier-Motzkin and CAD for projection and still suffer from doubly exponential complexity in the worst case.

Additionally, some works have utilized Gr\"obner bases in conjunction with CAD and Positivstellens\"atze theorems for existentially quantified formulas in NRA \cite{passmore2009combined, corzilius2012smt}. These approaches are also sound and complete, but again lead to doubly exponential algorithms. Furthermore, as discussed in the Introduction, some works have also proposed approaches to satisfiability checking for quantified formulas that avoid quantifier elimination. These include quantifier instantiation~\cite{ReynoldsTGKDB13,MouraB07,GeBT09} which is sound but incomplete, and the treatment of quantified first-order formulas as two-player games~\cite{BjornerJ15,farzan2016linear,MurphyK24} which is restricted to LRA. Finally, sKizzo \cite{Benedetti04} employs a symbolic Skolemization method for quantified boolean formulas, however this problem is fundamentally different from ours as we are working with the theory of reals.

\section{Preliminaries}\label{sec:prelims}

In this section, we define the syntax of linear real arithmetic (LRA) and non-linear real artihmetic (NRA) that we consider in this work. Since these are standard notions, we omit the formal definitions and assume that the reader is familiar with the semantics of LRA and NRA, the notion of satisfiability of a formula, etc. In what follows, we consider a finite set $V = \{x_1, x_2, \dots, x_n\}$ of distinct real-valued variables. 

\paragraph{Terms} The set of {\em terms} in LRA is defined via
\[ t \,\, ::= \,\, c \, \mid \, x \, \mid t_1 + t_2 \, \mid \, c \cdot t, \]
whereas the set of {\em terms} in NRA is defined via
\[ t \,\, ::= \,\, c \, \mid \, x \, \mid t_1 + t_2 \, \mid \, t_1 \cdot t_2, \]
where in both cases $c \in \mathbb{R}$ is a real-valued constant and $x$ is a variable in $V$. Hence, while in LRA a term can only be multiplied by a real-valued constant, NRA also allows multiplication of two terms, hence giving rise to polynomials.

\paragraph{Predicates} In both LRA and NRA, a {\em predicate} (sometimes also called a quantifier-free formula) is defined by the syntax
\[ p \,\, ::= \,\, t < 0 \, \mid \, t = 0 \, \mid \, p_1 \lor p_2 \, \mid \, p_1 \land p_2 \]
where $t$ is a term and $p_1$ and $p_2$ are also predicates. Note that logical negation $\neg$ is omitted in the above syntax, as it can be directly applied to the atomic predicates. 

\paragraph{Formulas} Finally, a {\em (quantified) formula} in both LRA and NRA is defined by the syntax
\[ \phi \,\, ::= \,\, \quantifier_1 x_1.\quantifier_2x_2.\dots\quantifier_n x_n. p,  \]
where $\quantifier_1, \quantifier_2, \dots,\quantifier_n \in \{\exists,\forall\}$ and $p$ is a predicate. For each $1 \leq i\leq n$, if $\quantifier_i = \exists$ we call it the {\em existential quantifier}, otherwise we call it the {\em universal quantifier}. In what follows, we assume that the reader is familiar with the notion of {\em satisfiability} of a formula.


\paragraph{Problem Statement} We consider the problems of checking satisfiability of formulas written in LRA and NRA:
\begin{enumerate}
	\item {\bf Problem 1: LRA Satisfiability.} Given a formula $\phi$ in LRA, check whether it is satisfiable.
	\item {\bf Problem 2: NRA Satisfiability.} Given a formula $\phi$ in NRA, check whether it is satisfiable.
\end{enumerate}
\section{Template-based Approach to Skolemization}\label{sec:skolemization}

In this section, we first recall the definitions of Skolem functions and present classical results from real algebraic geometry that illuminate their properties. We then use these properties as the foundation and justification for our template-based approach to Skolemization in our quantifier elimination procedure for LRA and NRA formulas. In particular, these results justify our choice for using linear and polynomial templates for Skolem functions. These will also be important in establishing the semi-completeness of our satisfiability checking algorithm in the following section.

\paragraph{Skolem Functions} Given a formula \(\phi(x_1,\cdots,x_n, y)\) in the first-order theory of reals, a {\em Skolem function} of \(\phi\) is defined as a function \(f : \mathbb{R}^n \rightarrow \mathbb{R}\) such that 
\[
\forall x_1.\cdots \forall x_n. \exists y. \, \phi(x_1,\cdots,x_n, y)  \iff 
\]
\[
\forall x_1.\cdots \forall x_n. \, \phi(x_1,\cdots,x_n, f(x_1,\cdots,x_n))
\]
Skolem functions allow us to remove existentially quantified variables from the formula and to replace them with functions over the preceding universally quantified variables.

\paragraph{Semi-Algebraic Sets} A set $S \subseteq \mathbb{R}^n$ is said to be {\em semi-algebraic}, if there exist two finite sets of polynomials \(P = \{p_1, \ldots, p_m\}\) and \(G = \{g_1, \ldots, g_k \}\) over $\mathbb{R}^n$ such that
\[
S = \left\{ x \in \mathbb{R}^n \, : \,
\begin{aligned}
    &p_1(x) \geq 0, \ldots, p_m(x) \geq 0, \\
    &g_1(x) = 0, \ldots, g_k(x) = 0
\end{aligned}
\right\}
\]


\paragraph{Semi-Algebraic Functions} Let \(A \subseteq \mathbb{R}^n\) and \(B \subseteq \mathbb{R}^m\) be semi-algebraic sets. A function \(f\colon A \rightarrow B\) is called {\em semi-algebraic} if its graph \( \Gamma(f) := \{(x, f(x)) \mid x \in A\} \) is a semi-algebraic set in \(\mathbb{R}^{n+m}\).

\begin{theorem}[\cite{scowcroft1988note}]
    Skolem functions for formulas in the first-order theory of reals are semi-algebraic.
\end{theorem}

Now, based on classical results from real algebraic geometry on semi-algebraic functions \cite{bochnak2013real}, the following properties hold true for Skolem functions in the first-order theory of reals:
\begin{compactitem}
    \item Skolem functions are piecewise continuous \cite[Proposition~5.20]{basu}.
    \item Skolem functions are bounded above by polynomials \cite[Proposition~2.6.2]{bochnak2013real}.
\end{compactitem}
Another noteworthy aspect of semi-algebraic functions to recall here: polynomial functions are also semi-algebraic.

The above properties of Skolem functions motivate us to consider polynomial functions as viable candidates for Skolem functions. In full generality, a Skolem function can be any function whose graph can be described by polynomial inequalities. However, the properties outlined above indicate that Skolem functions are piecewise continuous and bounded above by polynomials. Thus, searching for template-based polynomials as potential Skolem functions is a reasonable approach in our quantifier elimination procedure. This effectively prunes the search space for synthesizing Skolem functions for a given formula.

\section{Algorithm}\label{sec:algo}

We now present our algorithm for satisfiability checking for quantified formulas in LRA and NRA, which is the main contribution of our work. Our method is based on a novel procedure for quantifier elimination in LRA and NRA. Since our underlying algorithm for LRA and NRA is the same with only minor differences in certain steps, in what follows we provide a unified presentation for both theories and only highlight the differences that are specific to LRA or to NRA.

\paragraph{Goal} We aim for an algorithm for satisfiability checking which satisfies the following desirable properties:
\begin{compactitem}
	\item {\em Soundness}, which means that the output of our algorithm is guaranteed to be correct. That is, if the algorithm outputs that the formula is "satisfiable" (resp.~"unsatisfiable"), then it is indeed satisfiable (resp.~"unsatisfiable").
	\item {\em Semi-completeness}, which means that our algorithm is guaranteed to return an output whenever satisfiability or unsatisfiability of the quantified formula can be proved by a witnesses of a certain form. In our case, the class of witnesses will be linear/polynomial Skolem functions used in our template-based Skolemization approach.
	\item {\em Sub-exponential time and polynomial space complexity}, parametrized by the Skolem function templates size.
	\item {\em Witnesses of satisfiability for existentially quantified variables}, since computing these is important in many applications, as discussed in the Introduction.
\end{compactitem}

\paragraph{Algorithm Assumptions} In what follows, we assume that we are given a prenex normal form quantified formula
\[ \phi \,\, ::= \,\, \quantifier_1 x_1.\quantifier_2x_2.\dots\quantifier_n x_n. p \]
as defined in the Preliminaries, either in LRA or in NRA. Furthermore, we assume that the predicate $p$ is given in conjunctive normal form (CNF). The CNF assumption will be needed in Step~2 of our algorithm below. Finally, we assume that the user provides a maximal polynomial degree $D$ for Skolem function templates (formally defined below).

\begin{algorithm}[tb]
	\caption{Satisfiability checking for quantified LRA and NRA formulas}
	\label{alg:algorithm}
	\textbf{Input}: Quantified formula $ \phi \,\, = \,\, \quantifier_1 x_1.\quantifier_2x_2.\dots\quantifier_n x_n. p$ where $p$ is in CNF.\\
	\textbf{Parameter}: $D$, max polynomial degree of templates\\
	\textbf{Output}: (\texttt{SAT},existential model),~or~\texttt{UNKNOWN}~
	\begin{algorithmic}[1] 
		\STATE For every $i$ where $\quantifier_i \equiv \exists$, let $f_i$ be a polynomial over $x_1, \dots x_{i-1}$ of degree $D$ with unknown coefficients. \\
		\STATE Replace each occurrence of $x_i$ with $f_i$, and obtain  $\phi^{\textsc{univ}}~:=~\forall x_1 \dots x_m. p^{\textsc{univ}} \equiv \forall x_1 \dots x_m. \psi_1 \wedge \dots \land\psi_r$.
		\STATE Let $\phi^{\textsc{free}}:=\mathit{True}$.
		\FORALL{$\psi_i$}
		\STATE Convert $\psi_i$ into a polynomial entailment $\Phi \Rightarrow \psi$. 
		\STATE Apply Farkas, Handelman or Positivstellens\"atze Theorem to $\Phi \Rightarrow \psi$ in order to obtain $\Delta_{\psi_i}$.
		\STATE $\phi^{\textsc{free}} \leftarrow \phi^{\textsc{free}} \wedge \Delta_{\psi_i}$.
		\ENDFOR
		\IF{$\phi^{\textsc{free}}$ is satisfiable}
		\STATE $\mathit{model} := \mathit{getModel}(\phi^{\textsc{free}})$
		\STATE \textbf{Return} (\texttt{SAT,$\mathit{model}$})
		\ELSE 
		\STATE \textbf{Return} \texttt{UNKNOWN}
		\ENDIF
	\end{algorithmic}
\end{algorithm}

The rest of this section provides a detailed description of each of the three steps of our algorithm. We also illustrate each step on a running example. Algorithm \ref{alg:algorithm} shows the pseudocode of our satisfiability checking prpcedure.

\subsection{Step~1: Existential Quantifier Elimination}

In Step~1, the algorithm uses template-based Skolem functions to eliminate existential quantifiers from $\phi$.

For each existentially quantified variable $x_i$ in $\phi$, let $U_i$ denote the set of all universally quantified variables among $x_1,\dots,x_{i-1}$. Denote by $M_{i, D} = \{m_{i, 1}, \dots, m_{i, {k_i}}\}$ the set of all monomials of degree at most $D$ over variables in $U_i$.

The algorithm sets up a {\em template} for the Skolem function of $x_i$ as a symbolic polynomial of degree at most $D$ over the variables in $U_i$, i.e.~as a polynomial expression $f_i(U_i) = \sum_{j=1}^{k_i} c_{i, j} \cdot m_{i, j}$, where $c_{i, j}$'s are {\em template variables} that define polynomial coefficients. At this point, the values of template variables are unknown. The concrete real values that give rise to Skolem functions for each $x_i$ will be computed in Step~3 of the algorithm.

Finally, the algorithm constructs a purely universally quantified formula $\phi^{\textsc{univ}}$ from $\phi$ as follows. First, for each existentially quantified variable $x_i$, substitute each appearance of $x_i$ in $\phi$ by the Skolem function template $f_i(U_i)$ that was constructed for it. Then, remove all existential quantifiers since the existentially quantified variables have already been replaced by their Skolem functions. This procedure again results in a prenex normal formula of the form
\[ \phi^{\textsc{univ}} = \forall x_1.\dots\forall x_m. p^{\textsc{univ}}. \]
Since we are working over LRA or NRA and with polynomial Skolem functions, the predicate $p^{\textsc{univ}}$ is a boolean combination of polynomial (in)equalities. Furthermore, by the algorithm assumptions, the predicate $p^{\textsc{univ}}$ is in CNF.


\begin{example*}
	Consider the following quantified formula
	\begin{equation*}
		\begin{split}
			\phi \equiv &\forall x_1. \exists x_2. \exists x_3. \forall x_4 .\\&((x_4 - 1 < x_3 \lor x_2 \leq x_4) \land (x_4 > x_2 \lor x_3 \geq x_1)).
		\end{split}
	\end{equation*} 
	We use $D = 1$ for this example, i.e.~we are looking for linear Skolem functions. The above formula contains two existentially quantified variables $x_2$ and $x_3$. The variable $x_1$ is the only universally quantified variable preceding $x_2$ and $x_3$, hence $U_2 = U_3 = \{x_1\}$. Thus, the algorithm sets the following templates for Skolem functions for $x_2$ and $x_3$:
	$$f_2(x_1) = c_{2, 1} + c_{2, 2} \cdot x_1, \hspace{1em} f_3(x_1) = c_{3, 1} + c_{3, 2} \cdot x_1$$
	
	The algorithm substitutes these Skolem functions templates into $\phi$ and removes existential quantifiers to obtain:
	\begin{equation*}
		\begin{split}
			\phi^{\textsc{univ}} &\equiv \forall x_1. \forall x_4 .\\(&(x_4 - 1 < \mathbf{c_{3, 1} + c_{3, 2} \cdot x_1} \lor \mathbf{c_{2, 1} + c_{2, 2} \cdot x_1} \leq x_4) \land \\ &(x_4 > \mathbf{c_{2, 1} + c_{2, 2} \cdot x_1} \lor \mathbf{c_{3, 1} + c_{3, 2} \cdot x_1} \geq x_1))
		\end{split}
	\end{equation*}
\end{example*}


\subsection{Step~2: Universal Quantifier Elimination}

In Step~2, the algorithm eliminates universal quantifiers from the formula. Since the predicate $p^{\textsc{univ}}$ is a boolean combination of polynomial (in)equalities, the purely universally quantified formula $\phi^{\textsc{univ}}$ constructed in Step~1 can be viewed as a {\em polynomial entailment}. This is because it defines a boolean combination of polynomial (in)equalities that need to be satisfied for all universally quantified variable valuations. Hence, in order to remove universal quantifiers, we will use some of the classical results from real algebraic geometry \cite{bochnak2013real} that tackle the problem of {\em satisfiability of polynomial entailment}.


\paragraph{Satisfiability of Polynomial Entailment} Before proceeding further with our algorithm, we first define the problem of satisfiability of polynomial entailment and informally introduce the theorems used in our algorithm.

Consider a set $V = \{x_1, \ldots, x_r\}$ of real-valued variables, a set $\Phi = \{p_0 \bowtie 0, \dots, p_m \bowtie 0\}$ of polynomial inequalities over $V$, and a polynomial inequality $\psi = (p \bowtie 0)$ over $V$, where each  $\bowtie \in \{\ge, >\}$.
The {\em satisfiability of polynomial entailment} problem focuses on finding sufficient and necessary conditions for the universally quantified formula
\[ \forall x_1.\dots\forall x_r.\,\, (\Phi \Longrightarrow \psi) \]
to be satisfiable. Real algebraic geometry provides us with mathematical tools to reduce this problem to solving a system of polynomial inequalities, for which sub-exponential algorithms exist~\cite{grigor1988solving}. In particular, our algorithm will utilize the following results:
\begin{compactitem}
	\item \textbf{Farkas' Lemma.} This result provides a sound and complete method when both $\Phi$ and $\psi$ are linear. Given a system $\Phi$ of linear inequalities over $V$, Farkas' Lemma provides a set of necessary and sufficient conditions for $\Phi$ to be satisfiable and for $\Phi$ to entail a linear inequality $\psi$ over $V$. These conditions require that  $\psi$ can be written as non-negative linear combination of the inequalities in $\Phi$ and the trivial inequality $1 \ge 0,$ giving rise to an equivalent system of purely existentially quantified inequalities over the template variables as well as a set of new variables introduced by the Farkas' Lemma transformation.
	\item \textbf{Positivstellens\"atze Theorems.} These provide sound and semi-complete methods when either $\Phi$ or $\phi$ is non-linear. In the case when $\Phi$ is a system of linear inequalities over $V$ but the inequality $\psi$ is non-linear, we use Handelman's Theorem. Otherwise, if both $\Phi$ and $\psi$ contain non-linear inequalities, we use Putinar's Theorem. Both theorems provide a set of conditions for $\Phi$ to be satisfiable and to entail $\psi$, and their application gives rise to a system of purely existentially quantified inequalities over the template variables as well as a set of new variables introduced by the theorem transformation.
\end{compactitem}
We present the formal statements and the details of these theorems in the Appendix.

\paragraph{Universal Quantifier Elimination} Our algorithm uses the above theorems to eliminate universal quantifiers from the formula $\phi^{\textsc{univ}}$ constructed in Step~1. Since $p^{\textsc{univ}}$ is in CNF, we can write $p^{\textsc{univ}} \equiv \psi_1 \land \dots \land \psi_r$, with each $\psi_i \equiv \psi_{i, 1} \lor \dots \lor \psi_{i, {w_i}}$ and each $\psi_{i, 1}$ a polynomial inequality. 

The algorithm first converts each $\psi_i$ to the following equivalent form (if $w_i = 1$, then convert to $1 > 0 \Rightarrow \psi_{i, 1}$): 
\[
\lnot \psi_{i, 1} \land \dots \land \lnot \psi_{{i}, {w_i-1}} \Rightarrow \psi_{i, {w_i}}
\]
Then, depending on polynomial degrees in each $\psi_i$, the algorithm either applies Farkas' lemma if all polynomial degrees are $0$ and $1$, or Positivstellens\"atze theorems if higher degree polynomials are involved. For each $\psi_i$, the procedure results in a system of polynomial inequalities whose satisfiability implies satisfiability of $\psi_i$. Combining these systems together gives rise to a system of polynomial inequalities $\phi^{\textsc{free}}$ whose satisfiability implies satisfiability of $\phi^{\textsc{univ}}$. 


\begin{example*}[Continued]
	We first convert the formula $\phi^{\textsc{univ}}$ to the polynomial entailment form
	\begin{equation*}
		\begin{split}
			\phi^{\textsc{univ}} \equiv (&(x_4 - 1 \geq {c_{3, 1} + c_{3, 2}.x_1} \Rightarrow {c_{2, 1} + c_{2, 2}.x_1} \leq x_4) \land \\ &(x_4 \leq {c_{2, 1} + c_{2, 2}.x_1} \Rightarrow {c_{3, 1} + c_{3, 2}.x_1} \geq x_1))
		\end{split}
	\end{equation*}
	As all involved expressions are linear, the algorithm applies Farkas' lemma. Omiting the details of the translation, the algorithm obtains the following system of polynomial inequalities whose satisfiability implies satisfiability of $\phi^{\textsc{univ}}$:
	\begin{equation*}
		\begin{split}
			\phi^{\textsc{free}} \equiv &(-y_1 \cdot c_{3, 2} = - c_{2, 2}) \land \\&(y_1 \cdot (-1-c_{3, 1}) + y_2 = - c_{2, 1}) \land \\
			& (y_3 \cdot c_{2, 2} = c_{3, 2} - 1) \land (-y_3 = 0)\land \\&(y_3 \cdot (c_{2, 1}) + y_4 = c_{3, 1})
		\end{split}
	\end{equation*}
	The $c_{i,j}$'s are Skolem function template variables, whereas $y_1, \dots, y_4$ are fresh variables introduced by Farkas' lemma.
\end{example*}

\begin{remark*}[Distinction between LRA and NRA]
	Step~2 is the only part of our algorithm where the distinction between LRA and NRA arises. If the initial quantified formula $\phi$ is not expressible in LRA, then it must contain non-linear inequalities and so we cannot apply Farkas' lemma. Therefore, in the NRA case, we can only apply Positivstellens\"atze theorems.
\end{remark*}

\subsection{Step~3: Quantifier-free Satisfiability Checking}

In Step~3, we use an off-the-shelf SMT solver to check the satisfiability of the system of polynomial inequalities (i.e.~a quantifier-free formula) $\phi^{\textsc{free}}$ constructed in Step~2. If the SMT solver finds a solution, then we conclude that $\phi^{\textsc{free}}$ is satisfiable. By Step~2, this implies that $\phi^{\textsc{univ}}$ is satisfiable. Finally, by Step~1, this implies that the original quantified formula $\phi$ is satisfiable. Otherwise, our algorithm cannot prove satisfiability of the input formula. One can run the same procedure on $\neg\phi$ where its satisfiability is equivalent to $\phi$ being unsatisfiable.



\begin{example*}[Continued]
	Applying an SMT solver to the formula $\phi^{\textsc{free}}$ obtained before, we compute the following valuations for Skolem function template variables $c_{2, 1} = 0, c_{2, 2} = 1, c_{3, 1} = 1, c_{3, 2} = 1$. Hence, $f_2 = x_1$ and $f_3 = 1 + x_1$ define valid Skolem functions for $x_2$ and $x_3$ which prove the satisfiability of the original formula $\phi$. We
	provide a full example of the NRA settings in the Appendix.
\end{example*}

The following theorem establishes soundness, semi-completeness and a complexity bound for our algorithm. The soundness result states that, if our algorithm computes Skolem functions for the existential variables, then they are guaranteed to be correct. The semi-completeness result states that, if there exist polynomial Skolem functions of degree at most $D$ that witness formula satisfiability, then our algorithm is guaranteed to find them.

\begin{theorem}[Soundness, Semi-Completeness, Complexity, proof in the Appendix]\label{thm:lra}
	The algorithm is sound and semi-complete, both for quantified formulas in LRA and NRA. It runs in subexponential time and polynomial space, parametrized by the template polynomial degree $D$.
\end{theorem}


\paragraph{Witnesses for Existentially Quantified Variables} To conclude this section, we highlight that when the algorithm solves the system of polynomial constraints $\phi^{\textsc{free}}$ in Step~3, it computes concrete real values for template polynomial coefficients in Skolem functions constructed in Step~1. Hence, it computes {\em witnesses of satisfiability for existentially quantified variables}, which was one of the goals of our algorithm design that is important in many applications, as discussed in the Introduction.

\begin{table*}[htp]
	\centering
	\caption{Summary of the experimental results. Each row shows the performance of different tools on the set of benchmarks specified in the first column. For each solver, the \texttt{Solved} column shows the number of instances solved with the best results shown in bold, the \texttt{Avg. T.} column shows the average runtime (in seconds), and the \texttt{U.} column presents the number of unique instances solved by the tool (i.e.~instances that were solved only by that tool and no other tool). The timeout for each tool on each benchmark is set to 10~minutes.}
	\scriptsize{
	\texttt{
	\begin{tabularx}{\linewidth}{c|c|c|c|c|c|c|c|c|c|c|c|c|}
		\cline{2-13}
					& \multicolumn{3}{c|}{QuantiSAT} & \multicolumn{3}{c|}{Z3-uSk}	&  \multicolumn{3}{c|}{Z3-tSk}	&  \multicolumn{3}{c|}{CVC5-tSk} \\
		\cline{2-13}
					& 	Solved	& Avg. T. & U. &	Solved 	& Avg. T. & U.	&  	Solved & Avg. T. & U.	&  	Solved & Avg. T. & U. \\
		\cline{1-13}
		\multicolumn{1}{|c|}{Keymaera-LRA} 	&	\textbf{222}&	0.04	& 0		& 202 	&	0.01	&	0	& \textbf{222}	&	0.01	&	0	& 210 	&	0.01	& 0	\\
		\multicolumn{1}{|c|}{Mjollnir} 		&   \textbf{877}& 	18.48 	& 157	& 	476	&	14.99 	&	76	& 593			&	23.34 	&	2	& 397 	&	38.25	& 15\\
		\multicolumn{1}{|c|}{Keymaera-NRA} 	&	\textbf{503}&	0.06	& 22	& 481	&	0.01	&	0	& 483			&	0.03	&	2	& 481	&	0.01	& 0	\\
		\multicolumn{1}{|c|}{PolySynth} 	& 	\textbf{30} &	4.00	& 4		& 20  	&	3.44 	&	0	& 28  			&	11.16 	&	1	& 6   	& 13.74 	& 0	\\ 
		\cline{1-13}
	\end{tabularx}}}
	\label{table:experiments}
\end{table*}
\begin{figure*}[htp]
	\centering
	\subfigure{\includegraphics[scale=0.23]{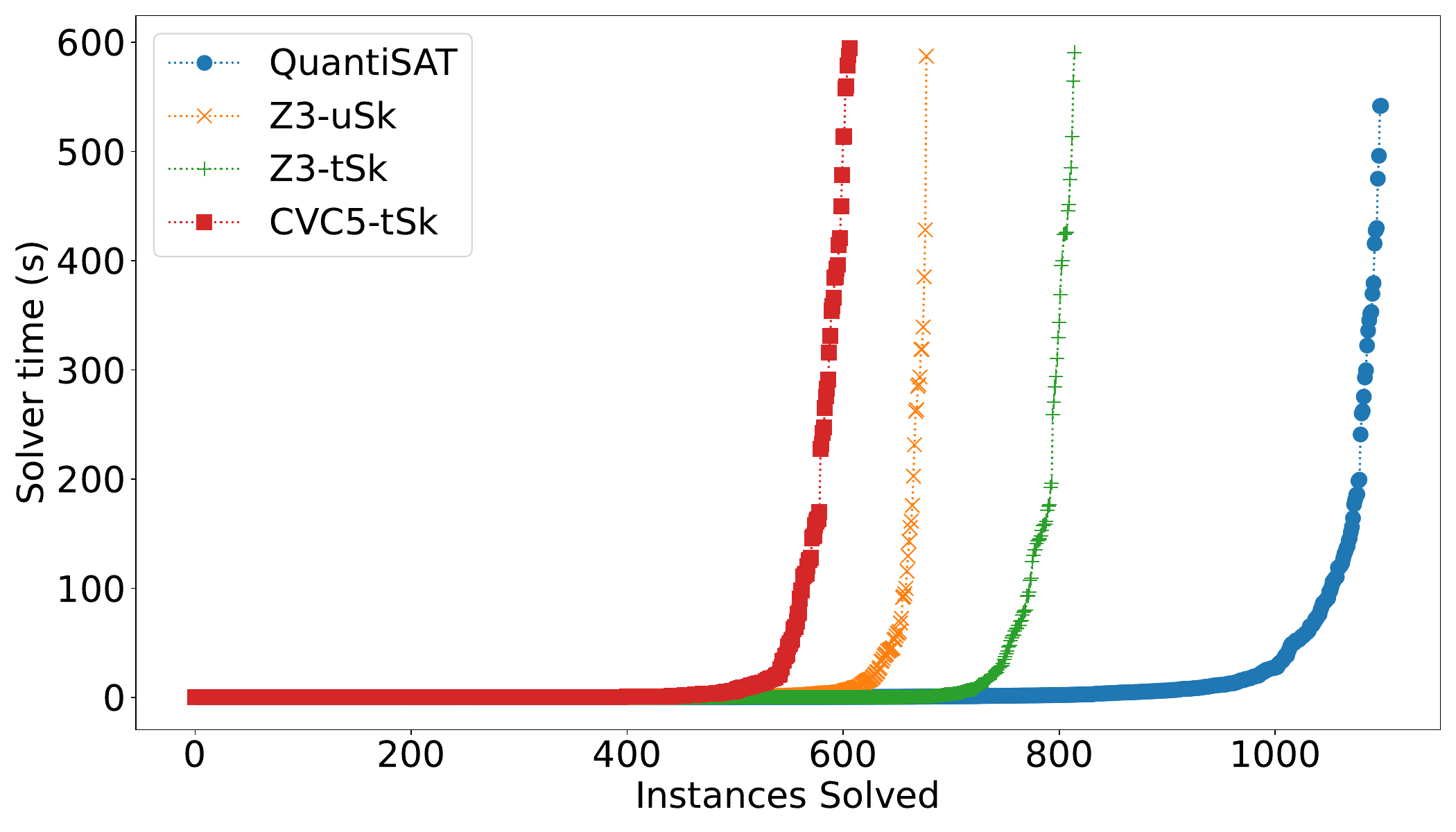}}\quad
	\subfigure{\includegraphics[scale=0.23]{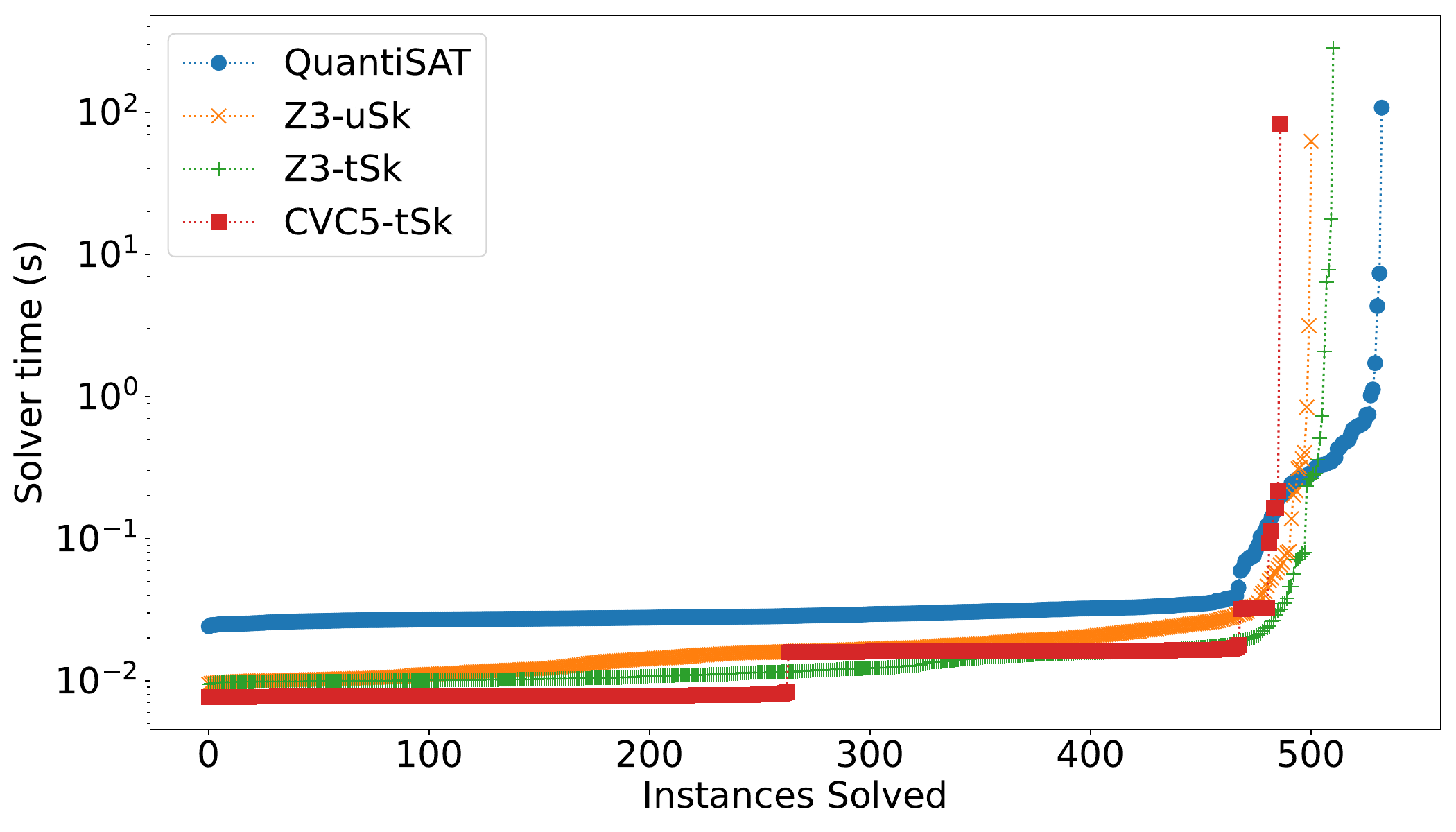}}
	\caption{Runtime distribution over the LRA (left) and NRA (right) benchmarks solved by each tool. The figure to the right is in logarithmic scale for better presentation. Generally, faster tools have lower curves and the more instances a tool solves, its curve will be more to the right. Since solving more instances is primary, tools with graphs to the right are preferable.}
	\label{fig:exp-times}
\end{figure*}
\section{Experimental Results}

We implemented a prototype of our method for satisfiability checking in a tool called QuantiSAT\footnote{https://doi.org/10.5281/zenodo.13341655}, and compared it against two state-of-the-art SMT solvers that support satisfiability checking for quantified LRA and NRA formulas. The goal of our experiments is to answer the following two research questions: (1)~How well does our method perform in comparison to the existing tools, in terms of the number of solved instances and runtime? (2) Is our method able to compute witnesses for existentially quantified variables?

\paragraph{Implementation Details} Our tool QuantiSAT is written in Python and it uses PolyHorn~\cite{chatterjee2024polyhorn} as a back-end tool for applying the Positivstellens\"atze theorems in Step~2 of the algorithm. It uses Z3~\cite{MouraB08} and MathSAT5 \cite{mathsat5} for solving the quantifier-free formulas derived in Step~3 of the algorithm. For the experiments, we used a Debian 12 machine with a 2.45GHz AMD EPYC 7763 CPU and 16 GB of RAM. In our experiments, we ran QuantiSAT with polynomial degrees for Skolem function templates equal to $D \in \{0,1,2\}$.


\paragraph{Benchmarks} We consider three benchmark suites of quantified formulas in LRA and NRA:
\begin{compactenum}
	\item The Keymaera benchmark suite, taken from SMT-COMP \cite{bobot18th}, contains 222 LRA formulas and 3813 non-linear formulas. However, many of the non-linear formulas are actually not expressible in NRA, as they contain the division operator which is not supported by NRA. Removing these results in 511 NRA formulas.
	\item The Mjollnir \cite{Monniaux10} benchmark suite consists of 3600 LRA formulas. This benchmark suite was used for the evaluation of the tool Mjollnir, which was later outperformed by Z3 \cite{BjornerJ15}, hence we do not include the Mjollnir tool in our results. 
	\item The PolySynth \cite{GoharshadyHMM23} benchmark suite consists of 32 NRA formulas. PolySynth is a program synthesis tool, which reduces the program synthesis problem to computing a satisfying assignment for a quantified formula in NRA. In our evaluation, we collected 32 quantified formulas that arise in their program synthesis procedure and used them to evaluate the effectiveness of our method and the baselines on these NRA benchmarks.
\end{compactenum} 
Since our method requires a quantifier-free part of the input formula to be provided in CNF (recall the assumptions in the Algorithm section), we first converted each quantifier-free part into CNF and then provided the CNF formulas as input to our tool and the baselines. The conversion time to CNF is not considered in the runtimes presented in Table~\ref{table:experiments}. 

\paragraph{Experimental Setup and Baselines} We compare QuantiSAT against two state-of-the-art SMT solvers Z3~\cite{MouraB08} and CVC5~\cite{cvc5}. The timeout for each tool on each benchmark is set to 10~minutes. Since we are not only interested in satisfiability checking but also in computing witnesses for existentially quantified variables that prove their satisfiability, in order to instruct Z3 and CVC5 to compute these witnesses, we provide them with two different variants of our benchmarks:
\begin{compactitem}
	\item {\bf Uninterpreted Skolemization.} In the first variant, we only require the baselines to compute {\em some witnesses} for existentially quantified variables, not necessarily being polynomial expressions. Hence, we replace each existentially quantified variable by an uninterpreted predicate over the preceding universally quantified variables. For Z3, we denote the resulting baseline by \texttt{Z3-uSk}. Since CVC5 does not support uninterpreted predicates, we could not evaluate it on this variant. Our goal here is to evaluate the effectiveness of our template-based Skolemization as opposed to general Skolem functions.
	
	\item {\bf Template Skolemization.} In the second variant, we ask our baselines to compute witnesses for existentially quantified variables {\em in terms of polynomials} over universally quantified variables. Hence, rather than only using uninterpreted predicates, in this variant we use the same polynomial Skolem function templates as in our QuantiSAT. We consider polynomial degrees $D \in \{0,1,2\}$, and count each instance as solved by the baseline if it can be solved for at least one of these three polynomial degrees. For Z3, we denote the resulting baseline by \texttt{Z3-tSk}. For CVC5, we denote the resulting baseline by \texttt{CVC5-tSk}. Our goal here is to evaluate the effectiveness of our quantifier elimination method based on Positivsellens\"atze theorems, as opposed to other quantifier elimination procedures implemented in these SMT solvers.
	
\end{compactitem}

\paragraph{Results on LRA Benchmarks} 
The first two rows of Table~\ref{table:experiments} summarize our experimental results on LRA benchmarks. It can be seen that average runtimes of all tools are comparable and quite small (differing by only a few seconds). The most important highlight of the table is the number of instances  and the unique instances solved by each tool, whereas runtimes of our and competing tools are the secondary aspect. We summarize our results on the LRA benchmarks below:
\begin{compactitem}
	\item {\bf Successful Instances.} (i) Our tool successfully solves all LRA benchmarks from the Keymaera benchmark suite, while \texttt{Z3-uSk} and \texttt{CVC5-tSk} fail on several cases. The fact that QuantiSAT and \texttt{Z3-tSk} solve all instances shows that the template-based approach to Skolemization provides an efficient and highly promising approach to quantifier elimination in LRA. (ii) On the Mjollnir benchmarks, QuantiSAT outperforms all the baselines with a gap of at least 284 instances, while solving 157 unique instances. On the other hand, \texttt{Z3-uSk} solves 76 unique instances that require complicated, e.g. heavily piecewise, interpretations for the Skolem functions. 
	\item {\bf Runtimes.} (i) The average runtime of all the considered tools is very small and negligible on the Keymaera benchmarks. (ii) On the Mjollnir benchmarks, the average runtime of our tool is lower than \texttt{Z3-tSk} and \texttt{CVC5-tSk} which means that applying Positivstellens\"atze is a crucial step for efficient quantifier elimination. Compared with \texttt{Z3-uSk}, although the average runtime of \texttt{Z3-uSk} is lower, our tool can solve in 2.45s the same number of instances that \texttt{Z3-uSk} solves in 600s. 
\end{compactitem}
Figure \ref{fig:exp-times}(left) shows the runtime distribution of each of the tool's runtimes over all the LRA benchmarks that were solved by that tool. QuantiSAT has the lowest and rightmost curve, which shows that it can solve more instances in less time compared to the baselines. These results show practical efficiency and applicability of our method to the satisfiability checking problem for quantified formulas in LRA.

\paragraph{Results on NRA Benchmarks}
The last two rows of Table~\ref{table:experiments} summarize our results on NRA benchmarks. Similar to the LRA results, the average runtimes of the tools are small and comparable. The most interesting distinction comes in the number of instances and unique instances solved by tools: 
\begin{compactitem}
	\item {\bf Successful Instances.}  (i) On the NRA benchmarks from the Keymaera benchmark suite, our tool outperforms the baselines by solving 98\% of the instances, which includes 22 instances uniquely solved by our tool. We believe that the high success rate is due to the strong semi-completeness guarantees provided by our method. (ii) On the benchmarks from the PolySynth benchmark suite, QuantiSAT outperforms all the baselines by solving 30 instances, including 4 unique ones.
	\item {\bf Runtimes.}  (i) On the Keymaera NRA benchmarks, the average runtime of QuantiSAT is higher than the baselines, however this is due to the longer runtimes required by the benchmarks solved only by our tool. QuantiSAT solves as many benchmarks as all the baselines in only 0.32 seconds. (ii) Comparing runtimes on the PolySynth benchmarks, \texttt{Z3-uSk} and \texttt{CVC5-tSk} have smaller average runtimes, however this is again due to the instances solved only by our tool. The only comparable baseline in terms of solved instances is \texttt{Z3-tSk}, whose average runtime is more than twice the runtime of QuantiSAT. 
\end{compactitem}
Figure \ref{fig:exp-times}(right) shows the runtime distribution of the NRA instances solved by each tool. 

\paragraph{Summary of Results} Based on the above discussions we conclude that, both on the LRA and the NRA benchmarks, QuantiSAT outperforms state-of-the-art and well-maintained tools such as Z3 and CVC5 on the number of solved instances when required to compute witnesses for existentially quantified variables. Furthermore, QuantiSAT is able to solve a significant number of new instances that other tools could not handle. Finally, this is achieved at improved average runtimes, as discussed above. All of this leads us to the conclusion that our method provides a significant step forward in satisfiability checking for quantified formulas in LRA and NRA, as well as in computing the witnesses. This is particularly due to the new quantifier elimination procedure that we propose. Given that the tool support for quantifier elimination (especially in NRA) is limited, we believe that our method provides important new ideas and breakthroughs in satisfiability checking for quantified formulas.

\section{Conclusion}\label{sec:conclusion}

We presented a novel method for satisfiability checking for quantified formulas in LRA and NRA. Our method is based on a novel and efficient quantifier elimination procedure. The method is sound, semi-complete, and runs in parameterized subexponential time and polynomial space. In contrast, previous sound and complete procedures have doubly-exponential time and at least exponential space complexity. Our method is also able to compute witnesses for existentially quantified variables, which is important for many applications. We implemented our method in a prototype tool called QuantiSAT. Our QuantiSAT outperforms two state-of-the-art SMT solvers on the number of solved instances and is able to prove a significant number of new instances that other tools could not handle. Hence, we believe that our method provides a significant step forward in satisfiability checking for quantified formulas in LRA and NRA as well as in efficient algorithms for quantifier elimination.

Our work opens several interesting future work directions. First, it would be interesting to consider further improvements to the template-based Skolemization method proposed in this work. These could include more general templates, such as {\em piecewise} linear and polynomial expressions. Second, it would be interesting to study the applicability of the template-based Skolemization technique to quantifier elimination in other theories, beyond LRA and NRA.

%

\bibliography{ref}
\newpage\clearpage
\appendix
\section{Background and Theorems on Satisfiability of Polynomial Entailment}

In this section, we provide an overview of the mathematical concepts and results on satisfiability of polynomial entailment, that we use in Step~2 of our algorithm







\subsection{Postivstellens\"atze}

Positivstellens\"atze are classical theorems from real algebraic geometry that provide an algorithmic approach to the problem of determining whether a polynomial is positive over a semi-algebraic set. Before presenting the specific results we used, it is important to note that there are various generalizations of Positivstellens\"atze, each applicable in different contexts. In our work, we focus on the versions that are most relevant and efficient for the development of our algorithm.

We begin with Farkas’ lemma, which characterizes the non-negativity of a linear expression over a set of linear inequalities.

\begin{theorem}[{Farkas' Lemma \cite{farkas1902theorie}}]{theorem}{Farkas}
    \index{Farkas' lemma}
    Consider a set $V = \{x_1, \ldots, x_r\}$ of real-valued variables and the following system $\Phi$ of equations over $V$:
        \begin{equation*}
            \Phi := 
            \begin{cases}
                a_{1,0} + a_{1,1} \cdot x_1 + \ldots + a_{1,r} \cdot x_r \ge 0\\
                \hfil \vdots\\
                a_{m,0} + a_{m,1} \cdot x_1 + \ldots + a_{m,r} \cdot x_r \ge 0\\
            \end{cases}.
        \end{equation*}
         When $\Phi$ is satisfiable, it entails a linear inequality $$\psi := c_0 + c_1 \cdot x_1 + \dots + c_r \cdot x_r \ge 0$$
        if and only if $\psi$ can be written as non-negative linear combination of the inequalities in $\Phi$ and the trivial inequality $1 \ge 0,$ i.e.~if there exist non-negative real numbers $y_0,\dots,y_m$ such that 
        $$
        \begin{matrix}
        c_0 = y_0 + \sum_{i=1}^k y_i \cdot a_{i, 0};\\
        c_1 = \sum_{i=1}^k y_i \cdot a_{i, 1};\\
        \vdots\\
        c_r = \sum_{i=1}^k y_i \cdot a_{i, r}.
        \end{matrix}
        $$
        Moreover, $\Phi$ is unsatisfiable if and only if $-1 \ge 0$ can be derived as above. \
\end{theorem}

Next, we present Handelman’s theorem, which extends Farkas’ lemma by providing a characterization of the non-negativity of a polynomial over a set of linear inequalities.

In order to present Handelman's theorem, we need to define the notion of a semi-group generated by a set of linear inequalities.

\begin{definition}[Semi-group generated by $\Phi$]
    \label{def:semi_group}
    \index{semi-group}
    Consider a set $V = \{x_1,\dots x_r \}$ of real-valued variables and the following system of linear inequalities over $V$:
    \begin{equation*}
    \Phi := 
    \begin{cases}
        a_{1,0} + a_{1,1} \cdot x_1 + \ldots + a_{1,r} \cdot x_r \bowtie_1  0\\
        \hfil \vdots\\
        a_{m,0} + a_{m,1} \cdot x_1 + \ldots + a_{m,r} \cdot x_r \bowtie_m 0\\
    \end{cases}
\end{equation*}
where $\bowtie_i \in \{>, \geq \}$ for all $1 \leq i \leq m$. Let $g_i$ be the left hand side of the $i$-th inequality, i.e. $g_i(x_1,\dots,x_r) := a_{i,0} + a_{i, 1} \cdot x_1 + \dots a_{i,r} \cdot x_r$. The \emph{semi-group} of $\Phi$ is defined as: 
   $$
\textstyle \monoid(\Phi) := \left\{ \prod_{i=1}^m g_i^{k_i} \mid m \in \mathbb{N} ~\wedge~\forall i ~\; k_i \in \mathbb{N} \cup \{0\} \right\}.
$$
We define $\monoid_d(\Phi)$ as the subset of polynomials in $\monoid(\Phi)$ of degree at most $d.$
\end{definition}
    
    \begin{theorem}[{Handelman's Theorem \cite{handelman1988representing}}]{theorem}{Handelman}\label{thm:handel}
        \index{Handelman's theorem}
    Consider a set $V = \{x_1, \dots, x_r \}$ of real-valued variables and the following system of equations over $V$:
         \begin{equation*}
            \Phi := 
            \begin{cases}
                a_{1,0} + a_{1,1} \cdot x_1 + \ldots + a_{1,r} \cdot x_r \geq  0\\
                \hfil \vdots\\
                a_{m,0} + a_{m,1} \cdot x_1 + \ldots + a_{m,r} \cdot x_r \geq 0\\
            \end{cases}.
        \end{equation*}
    If $\Phi$ is satisfiable, $\sat(\Phi)$ is compact, and $\Phi$ entails a polynomial inequality $g(x_1,\dots,x_r) > 0,$ then there exist non-negative real numbers $y_1,\dots y_u$ and polynomials $h_1,\dots,h_u \in \monoid(\Phi)$ such that:
    $$
    \textstyle g = \sum_{i=1}^u y_i \cdot h_i.
    $$
    \end{theorem}

Finally, we present Putinar’s Positivstellensatz, which extends both Farkas’ lemma and Handelman’s theorem by providing a characterization of the non-negativity of a polynomial over a set of polynomial inequalities.

    \begin{theorem}[Putinar's Positivstellensatz \cite{putinar1993positive}]\label{theorem_put_pos}
        \index{Putinar's positivstellensatz}
        Given a finite collection of polynomials $\{g,g_1,\dots,g_k \} \in \mathbb{R}[x_1,\dots,x_n]$, let $\Phi$ be the set of inequalities defined as
        $$
        \Phi: \{g_1 \geq 0,  \dots , g_k \geq 0 \}.
        $$
        If $\Phi$ entails the polynomial inequality $g > 0$ and there exist some $i$ such that $\sat(g_i \geq 0)$ is compact, then there exist polynomials $h_0, \dots, h_k \in \mathbb{R}[x_1,\dots,x_n]$ such that 
        $$
        \textstyle g = h_0 + \sum_{i=1}^{m} h_i \cdot g_i 
        $$
        and every $h_i$ is a sum of squares. Moreover, $\Phi$ is unsatisfiable if and only if $-1 > 0$ can be obtained as above, i.e.~with $g = -1$.
        \end{theorem}

The following theorem is a classical result from linear algebra that provides a characterization of sum of squares polynomials in terms of positive semidefinite matrices.

\begin{theorem}[{\cite[Theorem 3.39]{blekherman2012semidefinite}}]\label{thm:sos_semi}
	
	Let $\vec{a}$ be the vector of all $\binom{n+d}{d}$ monomials of degree less than or equal to $d$ over the variables $\{x_1, \dots, x_n\}.$ A polynomial $p \in \mathbb{R}[x_1,\dots,x_n]$ of degree $2 \cdot d$ is a sum of squares if and only if there exist a positive semidefinite matrix $Q$ of order $\binom{n+d}{d}$ such that $p = a^T \cdot Q \cdot a.$
\end{theorem}

\subsection{Algorithm Step~2 Details}

Using these theorems, we now present the mathematical foundations of Step~2 of our algorithm, which reduces the problem of checking satisfiability of a purely universally quantified formula to solving a system of polynomial inequalities. In particular, consider entailments of the form
\[
\exists C ~ \forall x ~ f_1 \geq 0 \wedge \ldots \wedge f_k \geq 0 \implies g > 0,
\]
where $f_1, \ldots, f_k, g$ are template-based polynomials, i.e., $f_i, g \in \mathbb{R}[C][x_1,\cdots,x_n]$, where $C$ denotes template variables. We can use the theorems mentioned above according to the following cases:

\begin{itemize}
\item \textbf{Case 1.} If $f_i$‘s and $g$ are linear polynomials, we use Farkas’ lemma to check if the entailment holds.
\item \textbf{Case 2.} If the $f_i$‘s are linear polynomials and $g$ is a polynomial, we use Handelman’s theorem to check if the entailment holds.
\item \textbf{Case 3.} If $f_i$‘s and $g$ are polynomials, we use Putinar’s Positivstellensatz to check if the entailment holds.
\end{itemize}

Without loss of generality, we present the reduction for Case~3. If we can express $g$ in the form of Equation~\ref{eq:sos}, then using Putinar’s Positivstellensatz, we know that the entailment holds.
\begin{equation}
g = h_0 + \sum_{i=1}^{m} h_i \cdot f_i
\label{eq:sos}
\end{equation}
where the $h_i$’s are sum of squares polynomials over the variables $x_1,\cdots,x_n$.

This reduces the problem of checking entailments to finding $h_0, \ldots, h_m$, which are sums of squares for a given degree $2 \cdot d$. Using Theorem~\ref{thm:sos_semi}, we know that this is equivalent to finding a positive semidefinite matrix $Q$ such that $h_i = a^T \cdot Q \cdot a$ for all $i$.

Next, we show how this problem can be reduced to a quadratic programming problem.

Since we are working with real numbers, to verify that Equation~\ref{eq:sos} holds, we only need to ensure that the coefficients on the LHS and RHS of the equation are equal. This eliminates the universal quantifier over $x_1,\cdots,x_n$.

When we equate both sides of the equation and expand the terms, we obtain a system of quadratic equations in the coefficients of the $h_i$’s and the template coefficients of the $f_i$’s and $g$. Additionally, we impose the constraint that $Q$ is positive semidefinite. This results in a quadratic programming problem that can be solved using standard solvers.

Similarly, for the other cases, the problem also reduces to solving polynomial inequalities.

For Case~1, we need to find non-negative real numbers $y_0,y_1,\dots y_u$ such that the following equation holds:

\[
g = y_0 + \sum_{i=1}^m y_i \cdot f_i.
\]

For Case~2, we need to find non-negative real numbers $y_1,\dots y_u$ and polynomials $h_1,\dots,h_u$ from the semi-group generated of a given degree $d$ by the $f_i$’s such that the following equation holds:

\[
g = y_0 + \sum_{i=1}^u y_i \cdot h_i.
\]

\paragraph{Strict and Non-strict Inequalities} Note that in our algorithm, we also work with the case when there are strict inequalities in hypothesis and conclusion of the entailment. In this case, the above mentioned steps can still be applied. Soundness of the algorithm for strict inequalities follows similarly as for non-strict inequalities. However, for completeness of the algorithm the above mentioned theorems are not enough. Fortunately, there are other results in literature that have been developed to handle non-strict inequalities \cite{GoharshadyHMM23,DBLP:conf/pldi/AsadiC0GM21}. We omit the details of these results for brevity. In particlar, one can use these methods along with Krivine–Stengle Positivstellensatz to handle strict and non-strict inequalities, which we present~next.

\paragraph{Krivine–Stengle Positivstellensatz}
This is the most general version of the Positivstellensatz. One can view the Krivine–Stengle Positivstellensatz \cite{krivine1964anneaux,stengle1974nullstellensatz} as a generalization of Putinar's Positivstellensatz, Handelman's theorem, and the real Nullstellensatz. Note that in the earlier versions of the Positivstellensatz we considered, had assumptions about compactness and the strictness of inequalities in the conclusions. However, to ensure the completeness of our algorithm, we require a version of the Positivstellensatz that can work without these assumptions.

Before we present the theorem, we need to recall the following definition of cone.

\begin{definition}[Cone]
    Let \( F = \{ f_1, \dots, f_m \} \) be a set of polynomials in \( \mathbb{R}[x_1, \dots, x_n] \). The cone generated by \( F \) is the set \( P(F) \) defined as follows:

    \[ 
        P(F) := \left\{ \sum_{\alpha \in \{0,1\}^m} \sigma_\alpha f_1^{\alpha_1} \cdots f_m^{\alpha_m} \mid \sigma_\alpha \in \Sigma^2[x_1, \dots, x_n] \right\},
    \]

    where \( \Sigma^2[X_1, \dots, X_n] \) is the set of all polynomials in \( \mathbb{R}[x_1, \dots, x_n] \) that are sums of squares. 
\end{definition}

\begin{theorem}[Krivine–Stengle Positivstellensatz { \cite{krivine1964anneaux,stengle1974nullstellensatz}}]
Let \( W \) be a semi-algebraic set in \( \mathbb{R}^n \) defined as follows:
\[
W = \{ x \in \mathbb{R}^n \mid g_1(x) \geq 0, \dots, g_s(x) \geq 0 \},
\]
and let \( P \) be the cone of \( \mathbb{R}[x_1, \dots, x_n] \) generated by \( g_1, \dots, g_s \). Let \( f \in \mathbb{R}[x_1, \dots, x_n] \). Then:

\begin{itemize}
    \item[(i)] \( \forall x \in W, \ f(x) \geq 0 \iff \exists m \in \mathbb{N}, \ \exists g, h \in P \text{ such that } f\cdot g = f^{2m} + h. \)
    
    \item[(ii)] \( \forall x \in W, \ f(x) > 0 \iff \exists g, h \in P \text{ such that } fg = 1 + h. \)
    
    \item[(iii)] \( \forall x \in W, \ f(x) = 0 \iff \exists m \in \mathbb{N}, \ \exists g \in P \text{ such that } f^{2m} + g = 0. \)
\end{itemize}
\end{theorem}

As in the steps mentioned earlier, we can fix a degree bound \( d \) for the cone \( P \) and \( m \), and then check if the entailment holds by finding suitable \( g, h \) up to degree \( d \) in the cone \( P \) such that the equation corresponding to the suitable case mentioned in the theorem are satisfied.

Specifically, the similar steps can be applied to generate the sum-of-squares polynomials, and as in previous cases, we need to check whether the coefficients of the polynomials on both sides of the suitable equation match. This results in a polynomial optimization step.

We would also like to highlight the similarity between the cone used in the Krivine–Stengle Positivstellensatz and the semigroup used in Handelman’s theorem, as well as the sum-of-squares used in Putinar's Positivstellensatz.

While this version of the Positivstellensatz is essential for ensuring completeness, our implementation utilizes Putinar’s Positivstellensatz and Handelman’s theorem, which offer greater practical efficiency.

In particular, we have also implemented a non-linear version of Handelman’s theorem, where the hypothesis involves non-linear polynomials, and the conclusion is likewise a polynomial. More precisely, we generated semigroup of polynomials upto a give a degree bound and then checked if the following equations holds:

\[
g = y_0 + \sum_{i=1}^u y_i \cdot h_i.
\]
where $h_i$'s are polynomials in the semigroup generated by the hypothesis polynomials and $y_i$'s are non-negative real numbers. This step can be seen as non-linear extension of Handelman's theorem, to the case when the hypothesis involves non-linear polynomials.

This version is sound and proves to be very efficient compared to the full strength of Putinar’s Positivstellensatz and the Krivine–Stengle Positivstellensatz. Notably, one can view this as a special case of the Krivine–Stengle Positivstellensatz, with the sums-of-squares in the cone being replaced by positive constants.

Now that we have presented the mathematical background, we are ready to prove the soundness, semi-completeness, and complexity of our algorithm.

\subsection{Proof of Theorem in Algorithm Section}

\begin{theorem*}[Soundness, Semi-Completeness, Complexity]
    The algorithm is sound and semi-complete, both for quantified formulas in LRA and NRA. It runs in subexponential time and polynomial space, parametrized by the template polynomial degree $D$.
\end{theorem*}

\begin{proof}
{\em Soundness.} Step~1 of the algorithm is sound since Skolemization by a general semi-algebraic function in quantified formulas in LRA and NRA is sound and complete by Theorem~1 in the main body of the paper~\cite{scowcroft1988note}. Hence, using polynomial Skolem functions is sound. Step~2 is sound by the soundness of Positivstellens\"atze theorems outlined above. Finally, Step~3 is sound since the procedures for solving systems of polynomial inequalities used by our SMT solvers are sound and complete~\cite{canny1988some,collins-cad75}. Thus, the whole algorithm is sound for satisfiability checking for quantified formulas in LRA and NRA.

{\em Semi-completeness.} The semi-completeness claim is that, if there exist polynomial Skolem functions of degree at most $D$ for each existentially quantified variable that together witness quantified formula satisfiability, then our method is guaranteed to find them. To prove this, observe first that in Step~1 the algorithm we fix a symbolic polynomial template of degree at most $D$ for each Skolem function. Hence, this step is semi-complete by definition. Next, in Step~2, we apply the Positivstellens\"atze theorems to check the satisfiability of polynomial entailments. In particular, depending on the case on the degree and strictness of the inequalities, we use Farkas' lemma, Handelman's theorem, Putinar's Positivstellensatz, or Krivine–Stengle Positivstellensatz as mentioned above. Without loss of generality, let us consider the case when hypothesis and conclusion consists of non-strict polynomial inequalities. More precisly, let us assume that the polynomial entailment is of the form:

\[
g_1 \geq 0 \wedge \ldots \wedge g_k \geq 0 \implies f \geq 0.
\]

In this case we apply the Krivine–Stengle Positivstellensatz.

In particular, we generate the cone of polynomials in the hypothesis $P(g_1,\dots,g_k)$, and check if the following equation holds:

\begin{equation}\label{eqn:complete}
    f\cdot g = f^{2m} + h, 
\end{equation}
    
where $m$ is a natural number less than $D$ and $h$ is a polynomial in the cone $P(g_1,\dots,g_k)$ of degree at most $D$. This step is semi-complete, as if there exists a solution for the given degree bound $D$ on the cone and $m$, then we are guaranteed to find it due to the Krivine–Stengle Positivstellensatz. Similarly, for other cases, we use Farkas' Lemma when hypothesis and conclusion consists of linear inequalities, Handelman's theorem when hypothesis consists of linear inequalities and conclusion is non-linear inequality, and other cases of Krivine–Stengle alongwith theorems presented in \cite{GoharshadyHMM23,DBLP:conf/pldi/AsadiC0GM21} when the conclusion and hypothesis consists of both strict and non-strict inequalities and all these steps are semi-complete.

Finally, in Step~3, we solve the system of polynomial inequalities over the template-variables to find their real valuations such that the equations obtained in Step~2 are satisfiable. This step is also semi-complete, as if there exists a solution, we are guaranteed to find it due to the completeness of the procedures for solving systems of polynomial inequalities used by our SMT solvers~\cite{canny1988some,collins-cad75}. Hence, the whole algorithm is semi-complete.

{\em Complexity.} The runtime complexity of Step~1 is polynomial when parametrized by the polynomial degree $D$, since it simply replaces each appearance of an existentially quantified variable by a symbolic polynomial Skolem function for that variable, with a polynomial degree of at most $D$. The runtime  complexity of Step~2 is also polynomial when parametrized by the polynomial degree $D$, as it syntactically applies the Positivstellens\"atze theorems outlined above which can be done in time polynomial in the size of the formula, when parametrized by $D$. Finally, Step~3 runs in subexponential time and polynomial space as it solves a system of polynomial inequalities over real-variables. This is because satisfiability checking in the first order existential theory of the reals can be done in subexponential time and polynomial space~\cite{grigor1988solving,canny1988some}. Hence, the whole algorithm runs in subexponential time and polynomial space, parametrized by the template polynomial degree $D$.
\end{proof}

\section{NRA example}
In this section, we provide an example for an NRA formula. We apply Putinar Positivstellens\"atze to find non-linear Skolem functions. Consider the following formula:

\begin{equation*}
	\begin{split}
		\phi \equiv &\forall x_1. \forall x_2. \exists x_3. (- x_1^2 - x_2^2 + 1 > 0 \lor x_1 < 10) \\&\land (x_1 < 0 \lor x_3 > x_1^2)
	\end{split}
\end{equation*}

Observe that linear Skolem functions cannot tackle this problem. We apply Putinar theorem with $D = 2$ and $d = 1$.

\paragraph{Step 1} In this example the only existential variable is $x_3$ with two preceding variables $x_1$ and $x_2$, so, $U_3 = \{x_1, x_2\}$. Thus the algorithm sets the following template function for $x_3$:

\begin{equation*}
	\begin{split}
		f_3(x_1, x_2) = &c_{2, 1} + c_{2, 2} . x_1 + c_{2, 3} . x_2 +\\& c_{2, 4} . x_1.x_2 + c_{2, 5} . x_1^2 + c_{2, 6} . x_2^2
	\end{split}
\end{equation*} 

Substituting this template into the initial formula $\phi$ we get:

\begin{equation*}
	\begin{split}
		\phi^{\textsc{univ}} \equiv &\forall x_1. \forall x_2. \\
					&\big(- x_1^2 - x_2^2 + 1 > 0 \lor x_1 < 10\big) \\
					&\land \big(x_1 < 0 \lor f_3(x_1, x_2) > x_1^2\big)
	\end{split}
\end{equation*} 

\paragraph{Step 2} In this step we convert the formula to entailment:

\begin{equation*}
	\begin{split}
		\phi^{\textsc{univ}} \equiv &\big(x_1^2 + x_2^2 \leq 1 \Rightarrow x_1 < 10\big) \\
		&\land \big(x_1 \geq 0 \Rightarrow f_3(x_1, x_2) > x_1^2\big)
	\end{split}
\end{equation*} 

Then, our algorithm applies Putinar with $d=1$. To this end, it generates new polynomials $h_0$, $h_1$, $h_2$, and $h_3$ where for each $i$:

\begin{equation*}
	\begin{split}
		h_i = &\eta_{i, 0} + \eta_{i, 1} . x_1 + \eta_{i, 2} . x_2 +\\& \eta_{i, 3} . x_1.x_2 + \eta_{i, 4} . x_1^2 + \eta_{i, 5} . x_2^2
	\end{split}
\end{equation*}

and using Theorem~\ref{thm:sos_semi}, ensures that each $h_i$ is actually a sum of squares (adds the required constraints to the input of Step 3). Then the algorithm generates the following equalities:

\begin{equation*}
	\begin{cases}
		-x_1 + 10 = h_0 + h_1.(-x_1^2 - x_2^2 + 1)\\
		f_3(x_1, x_2) - x_1^2 = h_2 + h_3.x_1
	\end{cases}
\end{equation*}

and equating both sides of the equations, we obtain a system of quadratic equations for the SMT solver.

\paragraph{Step 3} We input the system of equations obtained in Step 2 to an SMT solver and get the following result:

\begin{equation*}
	\begin{cases}
		f_3(x_1, x_2) = x_1^2 + 1\\
		h_0 = x_2^2 + (x_1 - \frac{1}{2})^2 + \frac{35}{4}\\
		h_1 = 1\\
		h_2 = 1 \\
		h_3 = 0
	\end{cases}
\end{equation*}

Hence $f_3 = x_1^2$ is a Skolem function that makes the formula $\phi$ true.

\end{document}